\documentclass[11pt]{article}%
\usepackage[in]{fullpage}%
\usepackage{mleftright}%
\usepackage{amsmath}%
\usepackage{amssymb}
\usepackage{mathtools}
\usepackage{booktabs}

\usepackage{enumerate}
\usepackage[boxruled,algosection,noend]{algorithm2e}
\usepackage{algorithmic}
\usepackage{graphicx}
\usepackage{xcolor}
\usepackage[inline]{enumitem}
\usepackage{scalerel}
\usepackage{needspace}
\usepackage{multirow}
\usepackage{multicol}
\usepackage{caption}%

\usepackage{hyperref}%
\definecolor{bluish-green}{HTML}{009E73}
\hypersetup{%
   breaklinks,%
   colorlinks=true,%
   urlcolor=[rgb]{0.25,0.0,0.0},%
   linkcolor=[rgb]{0.5,0.0,0.0},%
   citecolor=bluish-green,%
   filecolor=[rgb]{0,0,0.4}, anchorcolor=[rgb]={0.0,0.1,0.2}%
}

\usepackage{titlesec}%
\titlelabel{\thetitle. }%

\usepackage{changes}%
\definechangesauthor[name=siva, color=red]{s}
\definechangesauthor[name=makrand, color=orange]{m}
\definechangesauthor[name=cyrus, color=pink]{c}
\definechangesauthor[name=paul, color=green]{p}%

\usepackage{amsfonts, amsthm, amsmath}

\theoremstyle{plain}%
\newtheorem{theorem}{Theorem}[section]

\newtheorem{lemma}[theorem]{Lemma}
\newtheorem{proposition}[theorem]{Proposition}

\newtheorem*{defn:unnumbered}{Definition}

\numberwithin{equation}{section}

\newcommand{\nth}[1]{\ensuremath{#1^{\,\mathrm{th}}}}
\renewcommand{\th}{\ensuremath{\hphantom{}^{\,\mathrm{th}}}\xspace}

\newcommand{\HLinkShort}[2]{\hyperref[#2]{#1\ref*{#2}}}
\newcommand{\HLink}[2]{\hyperref[#2]{#1~\ref*{#2}}}
\newcommand{\HLinkPage}[2]{\hyperref[#2]{#1~\ref*{#2}%
      $_\text{p\pageref{#2}}$}}
\newcommand{\HLinkPageOnly}[1]{\hyperref[#1]{Page~\refpage*{#1}%
      $_\text{p\pageref{#1}}$}}

\newcommand{\HLinkSuffix}[3]{\hyperref[#2]{#1\ref*{#2}{#3}}}
\newcommand{\HLinkPageSuffix}[3]{\hyperref[#2]{#1\ref*{#2}%
      #3$_\text{p\pageref{#2}}$}}

\newcommand{\seclab}[1]{\label{sec:#1}}
\newcommand{\secref}[1]{\HLink{Section}{sec:#1}}

\newcommand{\lemlab}[1]{\label{lemma:#1}}
\newcommand{\lemref}[1]{\HLink{Lemma}{lemma:#1}}%

\newcommand{\tablab}[1]{\label{table:#1}}%
\newcommand{\tabref}[1]{\HLink{Table}{table:#1}}%

\newcommand{\proplab}[1]{\label{prop:#1}}

\newcommand{\propref}[1]{Proposition~\ref{prop:#1}}

\newcommand{\thmlab}[1]{{\label{theo:#1}}}
\newcommand{\thmref}[1]{\HLink{Theorem}{theo:#1}}

\providecommand{\eqlab}[1]{}%
\renewcommand{\eqlab}[1]{\label{equation:#1}}

\newcommand{\Eqref}[1]{\HLinkSuffix{Eq.~(}{equation:#1}{)}}

\IfFileExists{.latex_printer_friendly}{\def\GenPrinterVer{1}}{}%

\ifx\GenPrinterVer\undefined
   \IfFileExists{.latex_color}{\def\GenColorMath{1}}{}
\else
\fi

\ifx\GenColorMath\undefined
\else
\fi

\newcommand{\wt}{\widetilde}%

\newcommand{\polylog}{\mathrm{polylog}}

\usepackage{stmaryrd}%

\definecolor{blue25}{rgb}{0, 0, 11}

\newcommand{\remove}[1]{}

\DefineNamedColor{named}{AlgorithmColor}{rgb}{0.1,0,0.15}

\newcommand{\M}{\boldsymbol{M}}
\newcommand{\A}{\boldsymbol{A}}
\newcommand{\B}{\boldsymbol{B}}

\newcommand{\X}{\boldsymbol{X}}
\newcommand{\Y}{\boldsymbol{Y}}
\newcommand{\U}{\boldsymbol{U}}
\newcommand{\V}{\boldsymbol{V}}
\newcommand{\Z}{\boldsymbol{Z}}
\newcommand{\I}{\boldsymbol{I}}
\newcommand{\G}{\boldsymbol{G}}
\renewcommand{\H}{\boldsymbol{H}}
\newcommand{\T}{\mathrm{T}}
\renewcommand{\u}{\boldsymbol{u}}
\renewcommand{\v}{\boldsymbol{v}}
\newcommand{\w}{\boldsymbol{w}}
\newcommand{\x}{\boldsymbol{x}}
\newcommand{\y}{\boldsymbol{y}}
\newcommand{\e}{\boldsymbol{e}}
\renewcommand{\d}{\boldsymbol{d}}
\renewcommand{\c}{\boldsymbol{c}}
\renewcommand{\b}{\boldsymbol{b}}
\renewcommand{\a}{\boldsymbol{a}}
\newcommand{\bsig}{\boldsymbol{\Sigma}}

\newcommand{\F}{\mathbb{F}}
\newcommand{\Q}{\mathbb{Q}}
\newcommand{\R}{\mathbb{R}}
\newcommand{\uMv}{\mathsf{u^TMv}}

\makeatletter
\let\c@table\c@figure
\makeatother

\renewcommand{\epsilon}{\varepsilon}
\usepackage{booktabs}
\usepackage{color, colortbl}

\begin{document}

\title{Vector-Matrix-Vector Queries for Solving Linear Algebra, Statistics, and Graph Problems}

\date{\today}%

\author{%
   Cyrus Rashtchian%
   \thanks{Department of Computer Science \& Engineering, UC San Diego. %
      \url{crashtchian@eng.ucsd.edu} }%
   \and %
   David P. Woodruff%
   \thanks{Computer Science Department, Carnegie Mellon University.  
      \url{dwoodruf@cs.cmu.edu} }%
  \and %
  Hanlin Zhu%
  \thanks{Institute for Interdisciplinary Information Sciences, Tsinghua University. %
  	\url{zhuhl17@mails.tsinghua.edu.cn}}%
}

\date{\today}%

\maketitle

\begin{abstract}  
We consider the general problem of learning about a matrix through vector-matrix-vector queries. These queries provide the value of $\u^{\T}\M\v$ over a fixed field $\F$ for a specified pair of vectors $\u,\v \in \F^n$. To motivate these queries, we observe that they generalize many previously studied models, such as independent set queries, cut queries, and standard graph queries. They also specialize the recently studied matrix-vector query model. Our work is exploratory and broad, and we provide new upper and lower bounds for a wide variety of problems, spanning linear algebra, statistics, and graphs. Many of our results are nearly tight, and we use diverse techniques from linear algebra, randomized algorithms, and communication complexity.
\end{abstract}


\section{Introduction}

In the past few decades, there has been a significant amount of research on query-based algorithms, motivated by compressed sensing, streaming, sketching, distributed methods, graph parameter estimation, and property testing~\cite{canonne2015survey, eldar2012compressed, goldreich2017introduction, wang2015sublinear, woodruff2014sketching}. Most of this work focuses on local queries that only access a small portion of the unknown data at a time. For example, prior work on graph parameter estimation has considered \emph{degree} queries (which output the degree of a vertex $v$), \emph{edge existence} queries (which answer whether a pair $\{u,v\}$ forms an edge), and \emph{neighbor} queries (which provide the $i$\th neighbor of a vertex $v$). Not surprisingly, such queries have limited utility for certain problems. Even estimating the number of edges in a graph is known to require a polynomial number of edge existence, degree, and neighbor queries~\cite{f-sirvu-06, gr-aapg-08}. 

This has led researchers to consider queries that still reveal a small amount of information, while being more global in nature. For example, {\em bipartite independent set} queries (which indicate whether or not there is at least one edge between two disjoint sets of vertices) can be used to estimate the number of edges with only $\polylog(n)$ queries~\cite{bhrrs-eeiso-18, dlm-acssw-20}. Similarly, {\em cut} queries (which provide the number of edges crossing a graph cut) can be used to find the exact minimum cut in a graph~\cite{rubinstein2018computing, mukhopadhyay2019weighted}. Augmenting edge existence, degree, and neighbor queries with access to an edge sampling oracle (which provides a uniformly random edge) leads to elegant algorithms for estimating the number of certain subgraphs, e.g., triangles or cliques~\cite{assadi2018simple}, which was a major open problem (without edge sampling) up until a few years ago~\cite{elrs-actst-17, s-ssaat-15}. 

As the diversity of queries increases along with the range of applicable problems, it is natural to wonder whether there is a more general framework for understanding the power and limitations of query-based algorithms. In this work, we initiate the study of querying a matrix through bilinear forms, which generalizes the above mentioned queries and several more, sometimes with an $O(\log n)$ factor overhead. Formally, let $\M$ be an $n \times n$ matrix over a field $\F$. We consider vector-matrix-vector queries, which we call $\uMv$ queries for short. Given a pair of vectors $\u,\v \in \F^n$, these queries return the value of $\u^{\T}\M\v$ over $\F$. For graph applications, we often let the matrix $\M$ be the adjacency matrix of a graph. We later explain how to simulate standard graph queries with $\uMv$ queries. Allowing $\M$ to take values in other fields enables us to consider a greater variety of linear algebra, statistics, and data analytic problems.
The underlying field $\F$ will play an important role in our results, where working over $\F_2$ or $\R$ will change the query complexity of certain problems. We assume that the entries have $O(\log n)$ bit-complexity, and therefore, the output of one $\uMv$ query provides only $O(\log n)$ bits of information. We strive for algorithms using a subquadratic number of queries, which allows us to solve the problem without trivially learning the whole matrix. Unless we specify otherwise, we allow the queries to be randomized and adaptive.

From a practical point of view, algorithms based on $\uMv$ queries would most likely be useful in the context of specialized hardware or distributed environments. Computing a query only requires a weighted sum of entries of $\M$, and hence, it would be easy to execute in a massively parallel fashion. For example, if each processor handled a single row, then the local memory would be bounded by $O(n \log n)$ for storing $\u$ and $\v$. In a shared-nothing system, the number of communication rounds would be proportional to the number of queries. Similarly, in a streaming environment where single entries of $\M$ are changed at each step, the memory would be $O(\log n)$ times the number of queries. Working over a finite field $\F$ would reduce the memory overhead to $O(\log |\F|)$.

That being said, our focus is on the theoretical aspects of the $\uMv$ query model. We consider many problems, spanning linear algebra, statistics, and graph properties. Part of our motivation comes from finding algorithms that are query-efficient in the $\uMv$ model, while surpassing lower bounds for more restricted models. For example, we consider properties that depend on the whole matrix (e.g., having low rank, being unitary or doubly stochastic) or the entire graph (e.g., being a perfect matching or a star). As these are global properties, it is intuitively challenging to verify them using local queries without simply learning the whole matrix or graph. Overall, the $\uMv$ query model opens up many theoretical directions, and it facilitates new connections between linear algebra, randomized algorithms, and communication complexity.

We first provide an overview of the relationship between $\uMv$ queries and previously studied models. Then, we describe our results.

\subsection{Related work and other queries}
\seclab{related}

The $\uMv$ model provides a unifying lens and generalizes many previously studied queries.  

\begin{itemize}
\item {\bf Standard Graph Queries.} To gain intuition about $\uMv$ queries, we note that if $\M$ is the adjacency matrix of a graph, then a single query over a large field (e.g., $\Q$ or $\R$) provides the exact edge count. It is also easy to show that $O(\log n)$ $\uMv$ queries suffice to simulate degree, edge existence, neighbor, or edge sampling queries (see~\secref{local-graph} for details). Therefore, $\uMv$ queries achieve a variety of previous results with only an $O(\log n)$ factor overhead, such as estimating the number of cliques of different sizes~\cite{abgpry-stacs-18, assadi2018simple,elrs-actst-17,ers-ankcs-17, s-ssaat-15}, the number of stars~\cite{grs-csoss-11}, and the
minimum vertex cover~\cite{orrr-nosta-12}. 

\item {\bf Independent Set Queries.} Another line of work considers {\em independent set} oracles for graphs (which return whether a given set of vertices induces an independent set or contains at least one edge), in the context of estimating the number of edges in a graph~\cite{bhrrs-eeiso-18, clw-noeei-20, dl-fgrac-18, dlm-acssw-20}. Interestingly, {bipartite independent set} queries are known to be stronger than independent set queries~\cite{bhrrs-eeiso-18, clw-noeei-20}. 
Other variants of bipartite independent set queries, where
one of the sets is a singleton, have also been studied~\cite{bishnu2019inner,bkkr-csqtp-13,wly-imcvh-13}. While these algorithms are randomized and approximate, other work considers exact graph learning problems~\cite{abasi2019learning, alon2005learning, alon2004learning}.  When $\M$ is a binary matrix over a large enough field (e.g., $\Q$ or $\R$), then $\uMv$ queries generalize both independent and bipartite independent set queries by taking $\u$ and $\v$ to be indicator vectors for the sets. The power of the bipartite version motivates allowing $\u$ and $\v$ to differ in the $\uMv$ model.

\item {\bf Fine-Grained Complexity.} 
Independent set queries are partially motivated by studying the complexity of decision vs.~counting problems~\cite{dl-fgrac-18, dlm-acssw-20}. While we do not know of a natural use of $\uMv$ queries in this area, future work could consider using our algorithms for a similar complexity-theoretic reduction. Our model could also be extended to tensors, where queries are $k$-linear forms, analogous the generalization to $k$-partite independent set queries for counting $k$-cliques, which has applications to $k$-SUM and related problems~\cite{dlm-acssw-20}.

\item {\bf Cut Queries.} Another global graph query model considers {\em cut} queries (which provide the number of edges in a graph $G = (V,E)$ crossing a cut $(S, V \setminus S)$). It is known that $\wt O(n)$ cut queries suffice to exactly compute a minimum cut in a graph, and $\wt O(n^{5/3})$ queries suffice to compute an $s$-$t$ cut~\cite{rubinstein2018computing}. These results have also been extended to multigraphs~\cite{mukhopadhyay2019weighted}. We can directly simulate cut queries via indicator vectors $\u = \boldsymbol{1}_S$ and $\v = \boldsymbol{1}_{\{V \setminus S\}}$, when $\M$ is the adjacency matrix of the graph. As the $\uMv$ model is more general than cut queries, it an interesting open question whether a sublinear number of queries suffice for these problems.

\item  {\bf Matrix-Vector Queries.} 
A similar but more powerful query model considered by previous work involves matrix-vector queries~\cite{sun2019querying}. In this case, the queries return a vector of~$n$ values $\v^{\T}\M$ or $\M\v$ when given a vector $\v \in \F^n$. We study many of the same problems as this prior work. Certain problems, such as determining if a matrix is symmetric or diagonal, have constant query complexity in both models, even though $\uMv$ queries reveal much less information than matrix-vector queries. Previous work also considers lower bounds for the operator norm in the matrix-vector model~\cite{braverman2019gradient}, as well as the query complexity of computing PCA~\cite{simchowitzAR18}. Finally, we provide examples where matrix-vector queries are more powerful because there are lower bounds for $\uMv$ queries (see, e.g., \secref{reduction}). 

\item {\bf $\uMv$ Data Structures.} A complementary line of work considers the data structure complexity of $\uMv$ queries~\cite{ChakrabortyKL18, ChattopadhyayKLM18, DvirGW19, LarsenW17, natarajan2020equivalence}. More precisely, the goal is to preprocess $\M$ using a small amount space so that the value of $\u^{\T}\M\v$ can be obtained with a small query time (e.g., in the cell-probe model or natural restrictions of that model). Since there are connections between such data structures and challenging complexity theoretic problems (e.g., matrix rigidity, see~\cite{DvirGW19, natarajan2020equivalence}), it is an outstanding question to further explore whether our results have implications for $\uMv$  data structures or vice versa. 
\end{itemize}

\subsection{Our Results}
\seclab{results}

We provide new upper and lower bounds on the query complexity of various problems in the $\uMv$ model. \tabref{results} summarizes our results.  Many of the bounds are nearly tight: for some problems $O(1)$ queries suffice, and for others, either $\wt \Theta(n)$ or $\wt \Theta(n^2)$ are necessary and sufficient. We defer formal definitions to the relevant subsections. Here we highlight some interesting results.

\medskip \noindent {\bf General Techniques.} 
Querying the matrix with well chosen random vectors turns out to be a powerful algorithmic primitive that we employ often. In some cases, we use random indicator vectors to compare the number of ones in various submatrices (\thmref{permutation_upper} and \thmref{star}). Another technique is to choose random vectors $\u$, $\v$ whose entries are i.i.d.~and uniformly sampled from a field. If the matrix $\M$ satisfies certain properties, then $\u^{\T}\M\v$ will be nonzero with constant probability. We can prove this with the Schwartz-Zippel lemma if $\M$ is nonzero and the field has more than two elements; otherwise, for $\F_2$, we need a more elaborate analysis (\thmref{diagonal} and \thmref{identical}). We also use random Gaussian vectors (\thmref{all_one} and \thmref{identical}). Many of our lower bounds follow from a reduction to two-player communication complexity; we express the matrix as a function of two submatrices and show that the players can simulate the query algorithm to solve the communication problem  (\thmref{unitary_deterministic}, \thmref{maj}, \thmref{permutation_lower}).

\medskip \noindent {\bf Linear Algebra Problems.}  \secref{schatten} provides lower bounds for approximately computing many matrix norms, such as the trace norm, Frobenius norm, and operator norm (in general, we study Schatten $p$-norms; see~\secref{schatten} for the definition). To prove this result, we develop a general simulation result that allows us to establish lower bounds for adaptive $\uMv$ queries by reducing them to lower bounds for non-adaptive {\em entry-wise} queries. The key idea is that such a simulation result holds whenever the input matrix distribution is rotationally invariant (under row permutations). Then, we utilize known sketching lower bounds for matrix norms that identify a hard distribution that is rotationally invariant~\cite{li2019approximating}. 

On the upper bound side, we give constant-query algorithms for testing if a matrix is diagonal (\secref{diagonal}) or symmetric (\secref{symmetric}). While these algorithms are fairly straightforward, they exhibit the power of $\uMv$ queries to efficiently test for global properties of the matrix. 

We prove nearly-matching bounds for testing if a matrix is orthonormal (over $\R$) or unitary (over $\mathbb{C}$). The lower bound uses an encoding of information via the Hadamard matrix. 

\medskip \noindent {\bf Statistics Problems.}  
Turning to other matrix problems, we consider properties of one or more columns (our results also hold for rows, by symmetry of the query model). 
For example, \secref{all_one} and \secref{identical} provide nearly matching upper and lower bounds for testing if there is an all ones column or two identical columns. Many of our lower bound reductions require certain gadgets that seem to be new in the context of query complexity; for example, see our lower bounds for permutation matrices (\thmref{permutation_lower}). This also has led us to study negative entry detection in its own right, because a lower bound of $\Omega(n^2/\log n)$ from \thmref{negative} essentially provides the reason why certain results for binary matrices (e.g., graphs) cannot be generalized.

\medskip \noindent {\bf Graph Problems.}  Our upper bound on permutation matrices (\thmref{permutation_upper}) gives a constant-query algorithm over $\R$ for detecting whether a graph is a perfect matching. We also provide a constant-query upper bound over $\R$ for testing if a graph is a star on $n$ vertices (\thmref{star}). Both of these are global properties that would be difficult to verify using standard graph queries. They also complement previous results for learning hidden matchings or other structures using independent set queries~\cite{alon2005learning, alon2004learning}. As mentioned previously, simulating local graph queries with $O(\log n)$ $\uMv$ queries over $\R$ gives rise to a number of results on graph parameter estimation in the $\uMv$ model (see, e.g.,~\cite{abgpry-stacs-18, assadi2018simple,elrs-actst-17, ers-ankcs-17, grs-csoss-11, orrr-nosta-12, s-ssaat-15}). 

\medskip \noindent {\bf Organization.} We start with preliminaries in \secref{prelim}. We provide results for linear algebra problems in \secref{linear-algebra}, for statistics problems in \secref{statistics}, and for graph problems in \secref{graphs}. We conclude in \secref{conclusion}.

\begin{table}[htbp]
    \centering
    \caption{Our upper and lower bounds on the query complexity in the $\uMv$ model for $n \times n$ matrices and constant success probability. Results hold over any field unless stated otherwise.}
    \renewcommand{\arraystretch}{1.5}
    \begin{tabular}{lll}
        \rowcolor{gray!15}
        \multicolumn{3}{l}{\bf Linear Algebra Problems} \\
        Schatten $p$-norm &
        $\Omega(\sqrt{n})$ for $p \in [0,4)$, const.~factor approx.~over $\R$& \thmref{matrix-norm-lower}\\
        &
        $\Omega(n^{1-2/p})$ for $p\geq 4$, const.~factor approx.~over $\R$ & \thmref{matrix-norm-lower}\\
        Rank testing & $\Omega(k^2)$ to distinguish rank $k$ vs.~$k+1$ over $\F_p$ & \thmref{rank_F2} \\ 
        & $\Omega(n^{2-O(\varepsilon)})$ for $(1\pm \epsilon)$ approx.~over  $\R$, non-adaptive &   \thmref{rank_R} \\
         
        Trace estimation & $\Omega(n/\log n)$ and $O(n)$ for entries in $\{0, 1, 2, \ldots, n^3 \}$ &  \thmref{trace} \\
    
        Diagonal matrix & $O(1)$ & \thmref{diagonal} \\    
         
        Symmetric matrix & $O(1)$ & \thmref{symmetric} \\
         
        Unitary matrix & $\Omega(n/\log n)$ and $O(n)$ for randomized queries over $\mathbb{C}$ & \thmref{unitary_random} \\
        & $\Omega(n^2/\log n)$ for deterministic queries over $\mathbb{C}$ & \thmref{unitary_deterministic}\\
        \rowcolor{gray!15}
        \multicolumn{3}{l}{\bf Statistics Problems} \\
         
        All ones column & $\Omega(n/\log n)$ and $O(n)$ over $\mathbb{R}$ & \secref{all_one} \\
         
        Two identical columns & $\Omega(n)$  and $O(n \log n)$ over $\F_2$ & \secref{identical}\\ &  $O(n)$ over $\R$ & \thmref{identical} \\ 
         
        Column-wise majority & $\Theta(n^2)$ over $\F_2$ & \thmref{maj} \\
         
        Permutation matrix & $O(1)$ over $\R$ & \thmref{permutation_upper}\\
        
        & $\Omega(n)$ over $\F_2$ & \thmref{permutation_lower} \\
        
        Doubly stochastic matrix & $O(1)$ over  $\R$ & \thmref{double_stochastic_upper} \\
         
        Negative entry detection & $\Omega(n^2/ \log n)$ over  $\R$ & \thmref{negative} \\
        \rowcolor{gray!15}
        \multicolumn{3}{l}{\bf Graph Problems} \\
         
        Triangle detection & $\Omega(n^2 / \log n)$ & \thmref{triangle} \\
         
        Star graph & $O(1)$ over  $\R$ &  \thmref{star}\\
         
        Local graph queries & $O(\log n)$ over $\R$ &  \lemref{local}
    \end{tabular}
    \tablab{results}
\end{table}


\section{Preliminaries}
\seclab{prelim}

We use capital bold letters ($\A, \B, \X, \Y, \M, \ldots$) to represent matrices, lower-case bold letters ($\u,\v,\x,\y,\ldots$) to represent column vectors. We use non-bold lower-case letters ($x, y, \ldots$) to represent strings.
For a matrix $\M$, let $M_{ij}$ denote the entry in \nth{i} row and \nth{j} column. For a vector~$\v$, let $v_i$ denote the \nth{i} entry. For a string $x$, we use  $x_i$ to denote the \nth{i} entry. We use $\F$ to represent arbitrary fields, and use $\F_p$ to represent the finite field with $p$ elements where $p$ is prime, and $\R$ to denote the reals. 
We use $G = (V, E)$ to represent a simple graph, where $V$ denotes the set of vertices and $E$ denotes the set of edges. We query the adjacency matrix.

Some of our lower bounds use the communication complexity of $\textsc{Disjointness}$, where Alice has $x \in \{0, 1\}^n$, Bob has $y \in \{0, 1\}^n$, and they decide if there exists an index $i$ with $x_i = y_i = 1$. The randomized communication complexity is $\Omega(n)$~\cite{kalyanasundaram1992probabilistic, razborov1992distributional}. We also use the following result: if $x$ and $y$ contain exactly $n/4$ ones, then the randomized complexity is still $\Omega(n)$~\cite{bar2004information, haastad2007randomized}.

\section{Linear Algebra Problems}
\seclab{linear-algebra}

\subsection{Lower Bounds for Approximating Matrix Norms}
\seclab{schatten}
\seclab{reduction}

A distribution over matrices $\X \in \R^{n \times n}$ is {\em orthonormal and rotationally invariant} if (i) all rows of each $\X$ in the support are orthonormal and (ii) the distribution remains the same under any permutation of the rows of $\X$. We consider distributions over matrices $\M$ formed by fixing a diagonal matrix $\bsig$, sampling two matrices $\X$ and $\Y$ from orthonormal and rotationally invariant distributions, and letting $\M = \X \bsig \Y^{\T}$. At a high level, we are interested in algorithms for computing functions of the singular values $\bsig$, which remain invariant over matrices in such distributions. 

Our first goal is to prove a structural result relating $\uMv$ queries to entry-wise queries of~$\M$. Then, we use this reduction to prove new lower bounds. To do so, we utilize known streaming lower bounds, and we take advantage of the fact that these lower bounds are based on hard distributions that are orthonormal and rotationally invariant. Recall that $[s] = \{1,2,\ldots, s\}$ and that $\e_i \in \{0,1\}^n$ denotes the \nth{i} standard basis vector.
 
\begin{lemma}\lemlab{reduction-nonadaptive}
Let $\M = \X \bsig \Y^{\T}$ be a random $n \times n$ real-valued matrix, where $\bsig$ is diagonal, and $\X$ and $\Y$ are sampled from orthonormal and rotationally invariant distributions. Any $s \leq n$ deterministic, adaptive queries in the $\uMv$ model can be simulated by $s^2$ non-adaptive entry-wise queries to the values of $\e_i^{\T}\M \e_j$ for $i,j \in [s]$.
\end{lemma}
\begin{proof}
We proceed by induction on the number of queries $s \geq 1$. For the base case, consider a query $\u_1 \M \v_1$, where $\u_1, \v_1$ are arbitrary unit vectors. Observe that $\u_1^{\T}\X$ and $\Y^{\T}\v_1$ are random unit vectors, and moreover, they follow the same distribution as $\e_1^{\T}\X$ and $\Y^{\T}\e_1$, respectively. Since $\M = \X \bsig \Y^{\T}$, we see that the values of $\u_1^{\T}\M \v_1$ and $\e_1^{\T}\M \e_1$ are identically distributed as well. 

Suppose the lemma holds for any $s-1$ queries in the $\uMv$ model. Consider a sequence of $s$ queries 
\begin{equation}\eqlab{queries-reduction}
\u_1^{\T}\M \v_1, \ \u_2^{\T}\M \v_2,\ \ldots\ ,\  \u_{s}^{\T}\M \v_{s},
\end{equation}
for unit vectors $\u_i,\v_i$ for $i \in [s]$ that may depend adaptively on the previous queries. Assume without loss of generality that $\u_1,\ldots, \u_s$ and $\v_1, \ldots, \v_s$ are respectively linearly independent. For the final query vectors $\u_s$ and $\v_s$, decompose them as 
$$
\u_s = \a_s + \b_s 
\qquad \mbox{ and } \qquad 
\v_s = \c_s + \d_s,
$$
where 
$$
\a_s \in \mathsf{span}\{\u_1, \u_2, \ldots, \u_{s-1}\}
\qquad \mbox{ and } \qquad 
\c_s \in \mathsf{span}\{\v_1, \v_2, \ldots, \v_{s-1}\},
$$
and where $\a_s$ is orthogonal to $\b_s$, and $\c_s$ is orthogonal to $\d_s$.

Invoking the inductive hypothesis, this decomposition implies that $\a_s^{\T} \M \c_s$ can be simulated using  $\e_i^{\T}\M \e_j$ for $i,j \in [s-1]$. Furthermore, by the orthogonality assumptions, we have that $\b_s^{\T} \M \d_s$ follows the same distribution as $\e_s^{\T}\M \e_s$, even conditioned on the previous queries. 

It remains to argue about $\a_s^{\T} \M \d_s$ and $\b_s^{\T} \M \c_s$. We begin with the former, noting that the latter follows by a symmetric argument. Let $\w_1,\w_2,\ldots,\w_{s-1}$ denote an orthonormal basis for $\mathsf{span}\{\u_1, \u_2, \ldots, \u_{s-1}\}$. Considering the expansion of $\a_s$ in this basis, we observe that, by linearity, it suffices to simulate 
\begin{equation}\eqlab{basis-reduction}
\w_1^{\T}\M \d_s, \ \w_2^{\T}\M \d_s,\ \ldots\ ,\ \w_{s-1}^{\T}\M \d_{s}
\end{equation}
using only the information from $\e_i^{\T}\M \e_s$ for $i \in [s-1]$. To establish this, consider any vector $\w_i$ for $i \in [s-1]$. By assumption, $\X$ and $\Y$ are drawn from orthonormal and rotationally invariant distributions. Since $\w_1,\w_2,\ldots,\w_s$ form an orthonormal basis, we have that $\w_i^{\T} \X$ is a random unit vector following the same distribution as $\e_i^{\T} \X$. Moreover, by orthogonality, for any $i \geq 2$, the distribution of $\w_i^{\T} \X$ remains the same as $\e_i^{\T} \X$ even conditioned on
$$
\w_1^{\T}\X,\ \w_2^{\T}\X,\ \ldots\ ,\ \w_{i-1}^{\T}\X.
$$ 
An analogous argument implies that $\Y^{\T} \d_s$ follows the same distribution as $\Y^{\T} \e_s$, even conditioned on the previous queries. Therefore, we have that $\w_i^{\T} \M \d_s$ is identically distributed as $\e_i^{\T} \M \e_s$. Since this holds for all $i \in [s-1]$, the queries in \Eqref{basis-reduction} can be simulated by  $\e_i^{\T}\M \e_s$ for $i \in [s-1]$.  By symmetry, a similar result holds for simulating $\b_s^{\T} \M \c_s$. Therefore, we have shown that all $s$ deterministic queries in \Eqref{queries-reduction} can be simulated by the $s^2$ entry-wise non-adaptive queries to $\e_i^{\T}\M \e_j$ for $i,j \in [s]$, as desired.
\end{proof}

We use this structural result to prove lower bounds for computing certain matrix norms by applying sketching lower bounds due to Li, Nguyen, and Woodruff~\cite{li2019approximating}. For $p \in (0,\infty)$, the {\em Schatten $p$-norm} of a real matrix $\M \in \R^{n \times n}$ with singular values $\sigma_1, \ldots, \sigma_n$ is defined as 
$$
\| \M \|_p = \left(\sum_{i=1}^n \sigma_i^p\right)^{1/p}.
$$
By convention, the Schatten $0$-norm is the rank of the matrix, and the Schatten $\infty$-norm equals the largest singular value (a.k.a., operator norm). We have the following result for the $\uMv$ model.

\begin{theorem}\thmlab{matrix-norm-lower}
Let $\M \in \R^{n \times n}$ be a matrix. For any value $p \in [0,4)$, computing a constant-factor approximation to the Schatten $p$-norm of $\M$ requires $\Omega(\sqrt{n})$ $\uMv$ queries. For $p \geq 4$, computing a constant-factor approximation to the Schatten $p$-norm of $\M$ requires $\Omega(n^{1-2/p})$ $\uMv$ queries. Both results hold for randomized, adaptive queries with constant success probability.  
\end{theorem}

We sketch the proof of this theorem, which now follows directly from previous results. Before applying \lemref{reduction-nonadaptive}, we use Yao's principle~\cite{yao1977probabilistic} to show that it suffices to consider deterministic query algorithms for distributions over input matrices. Also, the query vectors can be taken to be unit vectors without loss of generality, as the algorithm can rescale the results. Then, we note that the previous lower bounds use hard distributions that are orthonormal and rotationally invariant~\cite{li2019approximating}. As a result, the distribution of matrices $\M$ satisfies the conditions of~\lemref{reduction-nonadaptive}. 

The previous results hold over the {\em bilinear sketching model}, where the sketches correspond to an $r \times n$ matrix $\U$ and an $s \times n$ matrix $\V$, and the goal is to approximate $\|\M\|_p$ up to a constant factor using $\U \M \V^{\T}$.  Applying \lemref{reduction-nonadaptive}, we see that any algorithm making $s$ queries in the $\uMv$ model corresponds to a bilinear sketch with both matrices being $s \times n$. Moreover, as the conclusion of the lemma only uses entry-wise queries, the corresponding matrices consist of the $s \times s$ identity matrix in the upper left-hand corner, while the rest of the matrix is all zeroes. The lower bound on bilinear sketches implies 
\begin{itemize}
\item $s^2 = \Omega(n)$ for approximating the Schatten $p$-norm with $p \in [0,4)$
\item $s^2 = \Omega(n^{2-4/p})$ for approximating the Schatten $p$-norm with $p \geq 4$.
\end{itemize}
Taking a square root leads to the bounds in \thmref{matrix-norm-lower}.

 The above provides separations between the $\uMv$ and matrix-vector models~\cite{sun2019querying}. Indeed, it is known that there exist non-trivial bilinear sketching matrices for approximating the Schatten $p$-norm whenever $p$ is an even integer. Denoting such sketching matrices as $\U$ and $\V$, it suffices for $\U$ and $\V^{\T}$ to each have $O(n^{1-2/p})$ rows~\cite{li2019approximating} to approximate the Schatten $p$-norm up to a constant factor. Observe that the Schatten $p$-norm of a matrix $\M$ is the same as the Schatten $p/2$-norm of the matrix $\M \M^{\T}$. Thus, if $p$ is an integer multiple of $4$, then in the matrix-vector model one can first compute $\U \M$ and then compute $\M^{\T} \V$, and then multiply these together to obtain $\U \M \M^{\T} \V$, where $\U$ and $\V$ are the corresponding sketching matrices for the Schatten $p/2$-norm. The total cost is $O(n^{1-4/p})$ queries in the matrix-vector model. 
 
 On the other hand, \thmref{matrix-norm-lower} implies that $\Omega(n^{1-2/p})$ queries are necessary in the $\uMv$ model, thus  providing a separation for integers $p \geq 4$ which are multiples of $4$. We also directly get an $\Omega(n)$ lower bound
 for approximating the operator norm up to a constant factor, using the $\Omega(n^2)$ lower bound bound for general sketches
 in \cite{lw16}. For recent work on actually finding the top eigenvector and solving a linear system in the matrix-vector model in the high accuracy regime, see \cite{braverman2019gradient}. 

\subsection{Rank Testing}

Given a matrix $\M \in \F^{n\times n}$, a natural problem is to determine the rank of $\M$. We first consider matrices over a finite field $\mathbb{F}_p$ for a prime $p$.

\begin{theorem} \thmlab{rank_F2}
    Given a matrix $\M \in \mathbb{F}_p^{n\times n}$ and an integer $k$, at least $\Omega(k^2)$ adaptive queries are necessary to decide the rank whether the rank of $\M$ is $k$ or $k+1$ with constant probability. 
\end{theorem}

\begin{proof}
    We reduce this problem to a communication complexity problem. Alice holds a matrix $\A \in \mathbb{F}_p^{n\times n}$ and Bob holds a matrix $\B \in \mathbb{F}_p^{n\times n}$, where $\M = \A + \B$ and $\mathsf{rank}(\M)  \in \{k, k+1 \}$.  Corollary 23 in \cite{li2014communication} implies that 
    the randomized communication complexity is $\Omega(k^2\log p)$ to determine whether the rank of $\M$ is $k$ or $k+1$. Alice and Bob can simulate the query algorithm using $O(\log p)$ bits of communication per query. Let $q(n,k)$ be the query complexity of this problem in the $\uMv$ model. Then $q(n, k)\log{p} = \Omega(k^2\log p)$, and we conclude that $q(n, k) = \Omega(k^2)$.
\end{proof}

Now consider the real-valued version of rank testing with $\M \in \mathbb{R}^{n\times n}$. 
It is known that if we want to compute the rank of $\M$ up to a factor of $(1\pm \epsilon)$, then this requires $\Omega(n^{2-O(\varepsilon)})$ space in the streaming model~\cite{assadi2017estimating}.
Assadi et. al.~\cite{assadi2017estimating} has shown that even for some special matrices of which the entries are only in $\{-1, 0, 1 \}$, there exists an $\Omega(n^{2-O(\varepsilon)})$ space lower bound for $(1+\varepsilon)$-approximation of the rank. Notice that for $\uMv$ queries, if we choose $\u = (1, 3, 3^2, \ldots, 3^{m-1})^{\T}$ and $\v = (1, 3^m, 3^{2m}, \ldots, 3^{m(m-1)})^{\T}$, then we can exactly reconstruct $\M$ using the value of $\u^{\T}\M\v$. Therefore, we assume that the matrix and the query vectors have integral values bounded by a polynomial in $n$.
Under this assumption, we prove the following theorem:

\begin{theorem} \thmlab{rank_R}
    Given a matrix $\M \in \mathbb{R}^{n\times n}$, if we restrict that the entry of query vectors can be chosen only from $\{0, 1, 2, \ldots, n^c\}$ for some constant $c$, then $\Omega(n^{2-O(\varepsilon)})$ non-adaptive queries are necessary to obtain a $(1+\varepsilon)$-estimation of $\mathsf{rank}(\M)$.
\end{theorem}
\begin{proof}
    Let $q(n)$ be the number of $\uMv$ queries sufficient to estimate the rank up to a factor of $(1\pm \epsilon)$. 
    Consider a streaming model with updates of the form $\u_1^{\T}\M\v_1$, $\u_2^{\T}\M\v_2$, $\ldots$, $\u_{q(n)}^{\T}\M\v_{q(n)}$, where $\u_i$ and $\v_i$ are the \nth{i} queries made in the $\uMv$ model for $i = 1,2,\ldots, q(n)$. We can store these queries using $O(q(n)\cdot \log n)$ bits of space (as the matrix and vector entries are polynomially bounded). Using the previous results of~\cite{assadi2017estimating}, we see that $\Omega(n^{2-O(\varepsilon)})$ bits of space are necessary. This implies that $q(n)O(\log n) = \Omega(n^{2-O(\varepsilon)})$, and hence, $$q(n) = \Omega(n^{2-O(\varepsilon)} / \log n) = \Omega(n^{2-O(\varepsilon)}).$$
\end{proof}

\subsection{Trace Estimation}

Estimating the trace of a matrix presents a simple problem where $\uMv$ queries are just as powerful as matrix-vector queries, even though the latter obtains much more information per query. Sun et. al.~\cite{sun2019querying} proves an $\Omega(n / \log n)$ lower bound in the matrix-vector model for trace estimation of symmetric matrix with entries in $\{0, 1, 2, \ldots, n^3\}$. Since the value of $\M\v$ contains all of the information of $\u^{\T}\M\v$, their result is also a lower bound for the $\uMv$ model. Of course, $n$ queries suffice to obtain all the diagonal elements of matrix $\M$, i.e.,
$
    \text{tr}(\M) = \sum_{i = 1}^n M_{ii} = \sum_{i = 1}^n \e_i^{\T}\M\e_i,
$
where $\e_i$ is the \nth{i} standard basis vector. Thus, for trace estimation, we obtain an $\Omega(n / \log n)$ lower bound and an $O(n)$ upper bound. We formalize this as the following theorem.

\begin{theorem} \thmlab{trace}
     Let $\M$ be an $n \times n$ matrix over $\R$ with entries in $\{0, 1, 2, \ldots, n^3 \}$. Assume the query vectors have entries in $\{0, 1, 2, \ldots, n^c \}$ for a constant $c > 0$. Computing a constant factor approximation to the trace $tr(\M)$ has query complexity between $\Omega(n / \log n)$ and  $O(n)$. 
\end{theorem}

\subsection{Deciding if a Matrix is Diagonal}
\seclab{diagonal}

In this section, we show over any field $\F$ that $O(\log (1/\varepsilon))$ queries suffice to test whether a matrix is diagonal with error probability at most $\varepsilon \in (0,1)$. To do so, we show that a single $\uMv$ query achieves constant success probability.

For each single-query test, we randomly and uniformly choose a subset $S$ of $[n] = \{1, 2, 3, \ldots, n\}$ with size $|S| = \frac{n}{2}$. We select a subset $G$ of size $|G|=2$ from $\F$. Construct the query vectors $\u$ and $\v$ as follows. For each $i \in [n]$, if $i \in S$, then let $u_i$ be randomly and uniformly sampled from $G$, and $v_i = 0$; otherwise let $v_i$ be randomly and uniformly sampled from $G$ and $u_i = 0$. 
If $\u^{\T}\M\v = 0$, then output `Success', otherwise output `Fail'. The whole algorithm outputs `Success' if and only if every test outputs `Success'. Now we formalize this as the following theorem and prove correctness.

\begin{theorem} \thmlab{diagonal}
    Let $\M$ be an $n \times n$ matrix over any field $\F$. Then with $O(\log(\frac{1}{\varepsilon}))$ queries, one can test whether $\M$ is a diagonal matrix with probability at least $1 - \varepsilon$.
\end{theorem}

\begin{proof}
    We show that for each query, if $\M$ is a diagonal matrix, then the test will always succeed; if not, then the test will fail with constant probability. Then by error reduction, $O(\log(\frac{1}{\varepsilon}))$ queries suffice to achieve error probability at most $\varepsilon$.
    
    For each query, we choose $\u$, $\v$ as the above algorithm describes. Therefore, 
    \begin{align*}
        \u^{\T}\M\v = \sum_{i\in S, j \in [n]-S} u_iv_jM_{ij}.
    \end{align*} 
    If $\M$ is diagonal, then $\u^{\T}\M\v = 0$ always holds. If $\M$ is not diagonal, then we claim that $\u^{\T}\M\v$ is non-zero with constant probability. In this case, there exists an off-diagonal element $M_{k\ell} \neq 0$ with $k \neq \ell$. With probability at least $\frac{1}{4}$, $k \in S$ and $\ell \in [n] \setminus S$ simultaneously. Conditioning on this event, let $t_i = \sum_{ j \in [n]-S} v_jM_{ij},$ and rewrite $\u^{\T}\M\v$ as 
    \begin{align*}
        \u^{\T}\M\v = \sum_{i\in S} u_i \sum_{ j \in [n]-S} v_jM_{ij} = \sum_{i\in S} u_i t_i.
    \end{align*}
    Let $T = \sum_{ j \in [n]\setminus S\setminus\{\ell\}} v_jM_{kj}$ and notice that
    \begin{align*}
        t_p = \sum_{ j \in [n] \setminus S} v_jM_{kj} =  v_\ell M_{k\ell} + \sum_{ j \in [n]\setminus S\setminus\{\ell\}} v_jM_{kj} = v_\ell  M_{k\ell} + T ,
    \end{align*}
    and $M_{k\ell} \neq 0$, $v_\ell M_{k\ell}$ has two different possible values. Moreover, at most one choice satisfies $v_\ell M_{k\ell} + T = 0$, and hence, $t_k \neq 0$ with probability at least $\frac{1}{2}$. 
    
    Now assume that $t_k \neq 0$. Let $R = \sum_{i\in S \setminus \{k\}} u_i t_i$, and notice that
    \begin{align*}
        \u^{\T}\M\v = \sum_{i\in S} u_i t_i = u_k t_k + \sum_{i\in S \setminus \{k\}} u_i t_i = u_k t_k + R.
    \end{align*}
    Since $t_k \neq 0$, by the same argument, $\u^{\T}\M\v \neq 0$ with probability at least $\frac{1}{2}$. Combining all of these events, the test fails with probability at least $\frac{1}{16}$, which completes the proof.
\end{proof}

\subsection{Deciding if a Matrix is Symmetric}
\seclab{symmetric}

Sun et. al.~\cite{sun2019querying} shows an $O(\log (\frac{1}{\varepsilon}))$ upper bound in the matrix-vector model to test whether an $n \times n$ matrix $\M$ is symmetric with probability $1-\epsilon$. We simulate their method in the $\uMv$ model by repeating the following process $O(\log (\frac{1}{\varepsilon}))$ times: choose random vectors $\u$,$\v$ and test whether $\u{^{\T}}\M\v = \v^{\T}\M\u$. Using the prior result~\cite{sun2019querying}, the error probability is at most $\varepsilon$. 
We formalize this as follows:

\begin{theorem} \thmlab{symmetric}
    Let $\M$ be an $n \times n$ matrix over any field $\F$. Then with $O(\log(\frac{1}{\varepsilon}))$ queries, one can test whether $\M$ is a symmetric matrix with probability at least $1 - \varepsilon$.
\end{theorem}

\subsection{Deciding if a Matrix is Unitary}
\seclab{unitary}

The results on query complexity in this subsection also apply for testing if a matrix is orthonormal over $\mathbb{R}$, since orthonormal is a special case of unitary.

\subsubsection{Randomized Queries}

Given an  $n \times n$ complex matrix $\M$, a single  matrix-vector query can determine whether $\M$  is unitary with probability one~\cite{sun2019querying}. Hence in the $\uMv$ model, $n$ randomized queries suffice, by obtaining the entries of the vector $\M\v$ using $\u = \e_i$ for $1 \leq i \leq n$.
Now we show that the $O(n)$ algorithm is nearly optimal by proving a lower bound $\Omega(n/ \log n)$ in the random case. 

\begin{theorem} \thmlab{unitary_random}
    Let $\M$ be an $n \times n$ matrix over $\mathbb{C}$. Then to determine whether $\M$ is a unitary matrix with a constant probability, the lower bound of query complexity is $\Omega(n / \log n)$ and the upper bound is $O(n)$.
\end{theorem}

\begin{proof}

Without loss of generality, let $n = 2^k$. We reduce the problem to  $\textsc{Disjointness}$. Suppose Alice has a string $x \in \{0, 1\}^n$, and Bob has a string $y \in \{0, 1\}^n$. Moreover, $\x$ and $\y$ both contain exactly $\frac{n}{4}$ ones, i.e. $\vert \{ i \in [n] \mid x_i = 1 \} \vert = \frac{n}{4}$ and $\vert \{ i \in [n] \mid y_i = 1 \} \vert = \frac{n}{4}$. Now Alice and Bob want to find whether there exists an index $i$ such that $x_i = y_i = 1$. The communication complexity of this problem is $\Omega(n)$~\cite{bar2004information, haastad2007randomized}. Now we show a protocol of the communication. First, let us recall one construction of a Hadamard matrix.

\begin{defn:unnumbered}
    Let 
    \begin{equation*}
        \H_1 = 
        \left[
        \begin{array}{c}
        1
        \end{array}
        \right ]
        , \text{ and }
        \H_{2^k} = 
        \left[
        \begin{array}{cc}
        \H_{2^{k-1}} & \H_{2^{k-1}} \\
        \H_{2^{k-1}} & -\H_{2^{k-1}}
        \end{array}
        \right ]
    \end{equation*}
    be a Hadamard matrix, 
    then we define $\G_{2^k}$ =  $\frac{1}{\sqrt{2^k}} \H_{2^k}$ for any $k \geq 1$.
\end{defn:unnumbered}

By the definition, $\G_{n/4}^*\G_{n/4} = \I_{n/4}$, which means $\G_{n/4}$ is a unitary matrix. Also, we denote the element of row $i$ and column $j$ of matrix $\G_{n/4}$ by $g_{i, j}$. Then Alice constructs an $n \times n$ matrix $\X$ with the following method.

Let $a_i$ denote the \nth{i} smallest position of string $x$ with value 1. For example, if $x = 00100010$, then $a_1 = 3, a_2 = 7$. Then Alice fills exactly $(n/4) \times (n/4)$ elements of matrix $\X$, i.e. 
\begin{equation*}
    \X_{a_i, a_j} = 
    \left\{
    \begin{array}{cc}
         g_{i, j} &, i \neq j \\
         g_{i, j} - 1 &, i = j,
    \end{array}
    \right. 
    \text{ where 1 $\leq$ $i, j$ $\leq$ $n/4$.}
\end{equation*}

Other elements of $\X$ are all 0s. Alice constructs another matrix $\X'$, which is the same as $\X$ except that
\begin{equation*}
    \X'_{a_i, a_j} = 
    \left\{
    \begin{array}{cc}
         - g_{i, j} &, i \neq j \\
         - g_{i, j} - 1 &, i = j,
    \end{array}
    \right. 
    \text{ where 1 $\leq$ $i, j$ $\leq$ $n/4$.}
\end{equation*} 

Let $\X_1 = \X$ and $\X_2 = \X'$. Bob uses the similar method to construct matrices $\Y_1$ and $\Y_2$ using his string $y$. If $x$ and $y$ are not intersected, i.e. there does not exist an index $k$, such that $x_k = y_k = 1$, then the four matrices $\M_{ij} = \X_{i} + \Y_{j} + \I$ are all unitary for $1 \leq i, j \leq 2$, where $\I$ is the identity matrix. However, if $x$ and $y$ intersects, then there exists an index $k$ such that $x_k = y_k = 1$. 
We argue that $\M_{ij}$ is not unitary because there exists $i, j \in \{ 1, 2 \}$ such that the element of \nth{k} row and \nth{k} column of matrix $\M_{ij}$ equals
\begin{align*}
    (-\frac{1}{\sqrt{n/4}} - 1) + (-\frac{1}{\sqrt{n/4}} - 1) + 1 = -1 - \frac{2}{\sqrt{n/4}} < -1. 
\end{align*}

Therefore, Alice and Bob can compute $\u^{\T}\M_{ij}\v$ by sending $\u^{\T}\X_i\v$ and $\u^{\T}\Y_j\v$, which take $O(\log n)$ bits by one communication. Assume that $q(n)$ queries can determine whether an $n \times n$ matrix is unitary, since $\textsc{Disjointness}$ requires $\Omega(n)$ bits, $q(n)O(\log n) = \Omega(n)$, which demonstrates that $q(n) = \Omega(n / \log n)$. 
\end{proof}

\subsubsection{Deterministic Queries}

For deterministic case, a trivial upper bound is $O(n^2)$ by retrieving all the entries of matrix $\M$ one by one. Now we show a strong lower bound $\Omega(n^2 / \log n)$, which demonstrates that the trivial algorithm is optimal up to the logarithmic factor. 

\begin{theorem} \thmlab{unitary_deterministic}
    Let $\M$ be an $n \times n$ matrix over $\mathbb{C}$. Determining whether $\M$ is a unitary matrix requires at least $\Omega(n^2 / \log n)$ queries in the deterministic $\uMv$ model.
\end{theorem}
\begin{proof}
    We reduce the problem to $\textsc{Disjointness}$. Without loss of generality, let $n = 2^k$. 
    The Hadamard matrix $\H_n$ contains $n^- = \frac{n(n-1)}{2}$ entries with value $-1$ and $n^+ = \frac{n(n+1)}{2}$ entries with value 1. We let $h_{i, j}$ denote the element of row $i$ and column $j$ in the matrix $\H_{n}$. 
Then, let $\Z$ be the $n \times n$ matrix defined as
\begin{equation*}
    Z_{i, j} = 
    \left\{
    \begin{array}{rl}
         -1, & \quad \mathrm{if\ \ } h_{i,j} = -1 \\
         0, & \quad \mathrm{if\ \ } h_{i, j} = 1.
    \end{array}
    \right. 
\end{equation*} 

Now Alice holds a string $x$ and Bob holds a string $y$, where $x, y \in \{ 0, 1 \}^{n^+}$. Each of the strings contains exactly $\frac{n^+}{2}$ 1s. In the deterministic case, it requires $\Omega(n^+) = \Omega(n^2)$ bits of communication to decide whether the two strings intersect. Alice constructs an $n \times n$ matrix $\X$ as follows. Initially, all entries of~$\X$ are zero. Then, linearly index the positions of 1s in the $n \times n$ matrix $\H_n$ simply by $1, 2, 3, \ldots, n^+$. For each~$i$, where $1 \leq i \leq n^+$, if $x_i = 1$, then Alice fills in 1 at position~$i$ in $\X$. Bob constructs a matrix $\Y$ using the string $y$ via the same method.  Let $\M = \frac{1}{\sqrt{n}}(\X+\Y+\Z)$. Notice that $x$ and $y$ do not intersect if and only if $\M$ is unitary. To exchange $\u^{\T}\X\v$ or $\u^{\T}\Y\v$ needs only $O(\log n)$ bits, so the lower bound is $\Omega(n^2 / \log n)$. 
\end{proof}

\section{Statistics Problems}
\seclab{statistics}

\subsection{All Ones Column} \seclab{all_one}

Let $\M \in \{ 0, 1 \}^{n \times n}$ be a binary matrix. Sun et. al.~\cite{sun2019querying} show a lower bound of $\Omega(n / \log n)$ for matrix-vector queries over $\R$ when restricting the entries in the query vector to lie in the set  $[n^c] = \{1, 2, \ldots, n^c\}$ for some constant $c$. This lower bound can be applied directly to the $\uMv$ model. The following theorem shows that this is tight up to  the logarithmic factor. 

\begin{theorem} \thmlab{all_one}
    Given a matrix $\M \in \{ 0, 1 \}^{n \times n}$ over $\R$, then $O(n)$ queries suffice to test whether there exists an all ones column in $\M$ with probability one.
\end{theorem}
\begin{proof}
    We construct a random vector $\u \in \R^n$, where each entry $u_i$ is independent and follows the standard Gaussian distribution. Let $\e_i$ denote the $n$ dimensional vector with \nth{i} entry 1 and all other entries 0s, and let $\e = \sum_{i=1}^n \e_i$ be the all ones $n$-dimensional vector. Also, let $\c_i$ denote the \nth{i} column of matrix $\M$. Since we have
    \begin{equation*}
        \u^{\T}\M\e_i = \u^{\T}
        \left[
        \begin{array}{cccc}
             \c_1 & \c_2 & \cdots & \c_n 
        \end{array}
        \right] 
        \e_i = \u^{\T}\c_i = \sum_{j=1}^n 
        u_jc_{ij},
    \end{equation*}
    if we compute the sum of all entries of $\u$, i.e.
$
        s = \sum_{i=1}^n u_i,
$
    then when $\c_i$ is an all ones column, all $c_{ij} = 1$, so $s = \u^{\T}\M\e_i$. Otherwise, 
    \begin{align*}
        s - \u^{\T}\M\e_i = \sum_{1 \leq j \leq n, c_{ij} = 0} u_j.
    \end{align*}
    The above quantity equals to 0 with probability 0, which means $s \neq  \u^{\T}\M\e_i$ with probability one.
    By querying $\u^{\T}\M\e_i$ for $1 \leq i \leq n$, and comparing the result to $s$, we can detect whether there is an all ones column with probability one, demonstrating an upper bound of $O(n)$ queries.
\end{proof}

\subsection{Identical Columns} \seclab{identical}

Let $\M \in \{0, 1\}^{n \times n}$. Rearrange $\M$ in the following way:
$$
    \M = 
    \left[
    \begin{array}{cccc}
    \c_1 & \c_2 & \cdots & \c_n
    \end{array}
    \right].
$$
We wish to determine whether there exists $i, j$, such that $1 \leq i < j \leq n$ and $\c_i = \c_j$.

We consider the lower bound on the query complexity over $\F_2$ first. 

\begin{theorem} 
    Let $\M \in \{0, 1 \}^{n \times n}$ be a binary matrix over $\F_2$. Let $\varepsilon$ be a real number such that $0 < \varepsilon < 1$ and $n \geq 2(1+ \log \frac{n^2}{\varepsilon})$, then $\Omega(n)$ queries are necessary to detect whether there exist two identical columns in $\M$ with  probability at least $1- \varepsilon$.
\end{theorem}

\begin{proof}
    We reduce this problem to $\textsc{Disjointness}$. Assume Alice has a string $x \in \{0, 1\}^{n-1}$, and Bob has a string $y \in \{0, 1\}^{n-1}$. Now Alice could construct a matrix $\X \in \{0, 1\}^{\frac{n}{2} \times n}$, where
    \begin{equation*}
        \X = 
        \left[
        \begin{array}{cc}
            \x^{\T} & 1 \\
            \a_1^{\T} & 1 \\
            \a_2^{\T} & 1 \\
            \vdots & \vdots \\
            \a_{\frac{n}{2}-1}^{\T} & 1
        \end{array}
        \right].
    \end{equation*}
    We denote the \nth{j} element of vector $\a_i$ as $a_{ij}$. For each $a_{ij}$, when $x_j = 1$, we let $a_{ij} = 1$; and when $x_j = 0$, we let $a_{ij}$ be a random variable drawn from a uniform distribution in $\{0, 1\}$. Bob constructs $\Y$ by the same method. Then let
    \begin{equation*}
        \M = 
        \left[
        \begin{array}{c}
        \X \\
        \Y
        \end{array}
        \right]
    \end{equation*}
    be an $n \times n$ matrix. If $x$ and $y$ intersect, then the corresponding column of $\M$ is all ones. Since the last column of $\M$ is also all ones, $\M$ contains two identical columns. If $x$ and $y$ do not intersect, then for every two columns, the probability that they are identical is at most $\frac{1}{2^{\frac{n}{2}-1}}$. By a union bound, the probability that there exist two identical columns is less than $\binom{n}{2} \times \frac{1}{2^{\frac{n}{2}-1}} \leq \frac{n^2}{2^{\frac{n}{2}-1}} \leq \varepsilon$, since $n \geq 2(1+ \log \frac{n^2}{\varepsilon})$. Alice and Bob must communicate $\Omega(n)$ bits, and sending $\u^{\T}\X\v$ and $\u^{\T}\Y\v$ each need only one bit over $\F_2$, so $\Omega(n)$ queries are necessary to detect two identical columns.
\end{proof}

For the upper bounds over $\F_2$ and $\mathbb{R}$, we have the following theorem.

\begin{theorem}\thmlab{identical}
    Let $\M \in \{0, 1\}^{n \times n}$ be a binary matrix. 
    \begin{itemize}
        \item $O(n \log(n/ \varepsilon))$ queries over $\F_2$ suffice to detect two identical columns with probability $1-\varepsilon$.
        \item $O(n)$ queries over $\mathbb{R}$ suffice to detect two identical columns with probability one.
    \end{itemize}
\end{theorem}

\begin{proof}
    We choose a random $n$-dimensional vector $\u$, where each $u_i$ is independent. Over $\mathbb{R}$, let $u_i$ be chosen from a standard normal distribution $N(0, 1)$; and over $\F_2$, let $u_i$ be chosen uniformly from $\{0, 1\}$. Notice that $n$ queries suffice to obtain $\langle \u, \c_i \rangle$ for $1 \leq i \leq n$, where $\c_i$ is the \nth{i} column of $\M$. If there are two identical columns $\c_i$ and $\c_j$, then $\langle \u, \c_i \rangle = \langle \u, \c_j \rangle$ always holds.
    
    Now we analyze the probability that $\langle \u, \c_i \rangle = \langle \u, \c_j \rangle$ holds for two columns that are not equal. For convenience, let $\v = \c_i - \c_j$. Since $\c_i \neq \c_j$, we know that $\v \neq \boldsymbol{0}$. Assume $v_k \neq 0$ for some index $k$ such that $1 \leq k \leq n$. When querying over $\mathbb{R}$, we have that
    \begin{align*}
        \langle \u, \c_i \rangle - \langle \u, \c_j \rangle
        = \langle \u, \c_i - \c_j \rangle = \langle \u, \v \rangle
        = \sum_{i=1}^n u_iv_i = u_kv_k + \sum_{i \neq k} u_iv_i.
    \end{align*}
    Since $u_k \sim N(0, 1)$ and $v_k \neq 0$, we have that $u_kv_k + \sum_{i \neq k} u_iv_i = 0$ with probability 0, which means that $\langle \u, \c_i \rangle \neq \langle \u, \c_j \rangle$ with probability one. Therefore, $O(n)$ queries suffice over $\R$ to detect identical columns with probability one.
    
    Working over the field $\F_2$,  we see that $u_kv_k + \sum_{i \neq k} u_iv_i = 0$ with probability $\frac{1}{2}$. This means that $\langle \u, \c_i \rangle = \langle \u, \c_j \rangle$ with probability $\frac{1}{2}$. If we choose $\log(n^2/\varepsilon) = O(\log(n / \varepsilon))$ independent vectors~$\u$, then this equality holds for every $\u$ with probability $\frac{\varepsilon}{n^2}$. Since there are $\binom{n}{2} \leq n^2$ pairs $(i, j)$, the overall error probability is less than $\frac{\varepsilon}{n^2} \cdot n^2 = \varepsilon$ by a union bound. Therefore, the query complexity in the $\uMv$ model over $\F_2$ is $O(n\log(n / \varepsilon))$.
\end{proof}

\subsection{Majority}

Given an $n\times n$ matrix $\M$ over $\F_2$, we consider computing the {\em column-wise majority} of $\M$. That is, for each column, we compute whether it contains at least $n/2$ ones or not. 
We prove that $\Theta(n^2)$ queries are necessary and sufficient, even for randomized algorithms.

\begin{theorem}\thmlab{maj}
Let $\M \in \F_2^{n\times n}$ be a binary matrix. Computing the column-wise majority of~$\M$ requires $\Omega(n^2)$ queries, even for constant success probability.
\end{theorem}
\begin{proof}
    We reduce this problem to $\textsc{Disjointness}$. Assume Alice has $n$ binary strings of length $n$, i.e. $x_1, x_2, \ldots, x_n$, each of which contains exactly $\frac{n}{4}$ 1s. Bob has  $n$ binary strings of length $n$, i.e. $y_1, y_2, \ldots, y_n$, each of which contains exactly $\frac{n}{4}$ 1s as well. We define $f : \{0, 1\}^n \times \{0, 1\}^n \rightarrow \{0, 1\}$ as follows:
    \begin{equation*}
    f(x, y) = 
    \left\{
    \begin{array}{cl}
         0 , & x \text{ and } y \text{ have non-empty intersection,}  \\
         1 , &\text{otherwise}.
    \end{array}
    \right. 
\end{equation*} 
    
    By a direct sum theorem in communication complexity, $\Omega(n^2)$ bits of communication are required to decide $(f(x_1, y_1), f(x_2, y_2), \ldots, f(x_n, y_n))$ simultaneously~\cite{molinaro2013beating}. Let $\x_i$ be the corresponding $n$-dimensional column vector of string $x_i$. Also, let $\y_i$ be the corresponding column vector of string $y_i$.   Alice and Bob construct matrices $\X$ and $\Y$, where
    \begin{equation*}
        \X = 
        \left[
        \begin{array}{cccc}
        \x_1 & \x_2 & \cdots & \x_n
        \end{array}
        \right]
\qquad \mbox{ and } \qquad
        \Y = 
        \left[
        \begin{array}{cccc}
        \y_1 & \y_2 & \cdots & \y_n
        \end{array}
        \right].
    \end{equation*}
    Let $\M = \X + \Y$. Then $x_i$ and $y_i$ intersect if and only if the majority of \nth{i} column of $\M$ is 0 since the elements are over $\F_2$. Furthermore, $\u^{\T}\M\v$ can be computed by the communication of $\u^{\T}\X\v$ and $\u^{\T}\Y\v$, each communication requiring one bit. Thus, the number of queries needed to decide the majority of every column is $q(n) = \Omega(n^2)$.
\end{proof}
    
    

\subsection{Permutation Matrix}

A matrix $\M \in \{0, 1\}^{n\times n}$ is a {\em permutation matrix} if each column and each row contains exactly one entry equal to 1. We consider the query complexity over both $\R$ and $\F_2$, which are very different. 

We observe checking if a graph $G$ is a perfect matching is equivalent to checking if the adjacency matrix is a permutation matrix. This also holds for the bipartite version: for a graph on $2n$ vertices, let $M_{ij} = 1$ when the \nth{i} vertex on the left is connected to the \nth{j} vertex on the right.

The following theorem states that $O(1)$ queries suffice over the reals to check whether $\M$ is a permutation matrix with constant probability.

\begin{theorem}\thmlab{permutation_upper}
    Let $\M \in \{0, 1\}^{n \times n}$ be a binary matrix over $\mathbb{R}$. Then, 
    $O(1)$ queries suffice to check whether $\M$ is a permutation matrix with constant probability.
\end{theorem}

\begin{proof}
    Using a single query $\u = \v = (1, 1, \ldots, 1 )^{\T}$, we first verify that $\M$ contains exactly $n$ ones. Assume this holds. Also, assume without loss of generality that $n$ is even.
    
    We first describe an algorithm to test with constant probability whether each column contains a single one. Reversing the roles of columns and rows will establish the same for rows. The algorithm repeats the following process a constant number of times. Randomly select a subset $A \subseteq [n]$ of exactly $n/2$ columns. Let $\u$ be the all ones vector, and let $\v = \mathbf{1}_A$ and $\v' = \mathbf{1}_{A \setminus [n]}$ be the indicator vectors for $A$ and its complement. Reject if either $\u^{\T}\M \v \neq n/2$ or $\u^{\T}\M \v' \neq n/2$.      
    
    If $\M$ is a permutation matrix, then $\u^{\T}\M \v = \u^{\T}\M \v' = n/2$ holds. 
    If $\M$ is not a permutation matrix, there must be a pair of columns (or rows), one with all zeros, and one with at least two ones.  Suppose column $c$ contains all zeros, and column $c'$ contains at least two ones. 
    With constant probability in choosing $A$, we have $c \in A$ and $c' \notin A$ or vice versa. Conditioned on this, we claim that either $\u^{\T}\M \v \neq n/2$ or $\u^{\T}\M \v' \neq n/2$ with constant probability as well. 
    
    To see this, consider randomly partitioning the $n-2$ columns (excluding $c$ and $c'$) into two groups of size $\frac{n}{2} - 1$.  Let $s_1$ and $s_2$ be the number of ones in Groups 1 and 2, respectively. Without loss of generality, assume $s_1 \leq s_2$. Now, consider adding $c$ and $c'$ to the two groups, conditioned on them being separated. If $c'$ is in Group 2, then Group 2 will have more ones than Group 1. Thus, one of the groups must not have $n/2$ ones, and our algorithm rejects with constant probability.
\end{proof}

Interestingly, the query complexity depends on the field. If $\M \in \{0, 1\}^{n \times n}$ is over $\F_2$, then $O(1)$ queries are far from enough.

\begin{theorem} \thmlab{permutation_lower}
    Let $\M \in \F_2^{n \times n}$ be a matrix. Then, $\Omega(n)$ queries are necessary to determine whether~$\M$ is a permutation matrix with constant probability.
\end{theorem}

\begin{proof}
    We reduce this problem to $\textsc{Disjointness}$. Alice holds a string $x \in \{0,1\}^n$ and Bob holds a string $y \in \{0,1\}^n$. Now Alice constructs a $3n\times 3n$ matrix 
    \begin{gather*}
        \A = 
        \left[
        \begin{array}{cccc}
             \A_1 & \boldsymbol{0} & \cdots & \boldsymbol{0} \\
             \boldsymbol{0} &  \A_2 & \cdots & \boldsymbol{0} \\
             \vdots & \vdots & \ddots & \vdots \\
             \boldsymbol{0} & \boldsymbol{0} & \cdots & \A_n 
        \end{array}
        \right]
    \qquad \mbox{ where } \qquad 
        \A_i = 
        \left\{
        \begin{array}{cc}
           \left[
           \begin{array}{ccc}
                0 & 0  & 1 \\
                1 & 0  & 0\\
                0 & 1  & 0
           \end{array}
           \right]  & \mbox{if } x_i = 0 \\ \\
            \left[
           \begin{array}{ccc}
                1 & 0  & 0 \\
                0 & 1  & 0\\
                0 & 0  & 1
           \end{array}
           \right] & \mbox{if } x_i = 1.
        \end{array}
        \right.
    \end{gather*}
    Bob constructs a $3n\times 3n$ matrix 
    \begin{gather*}
        \B = 
        \left[
        \begin{array}{cccc}
             \B_1 & \boldsymbol{0} & \cdots & \boldsymbol{0} \\
             \boldsymbol{0} &  \B_2 & \cdots & \boldsymbol{0} \\
             \vdots & \vdots & \ddots & \vdots \\
             \boldsymbol{0} & \boldsymbol{0} & \cdots & \B_n 
        \end{array}
        \right]
    \qquad \mbox{ where } \qquad 
        \B_i = 
        \left\{
        \begin{array}{cc}
           \left[
           \begin{array}{ccc}
                0 & 0  & 0 \\
                0 & 0  & 0\\
                0 & 0  & 0
           \end{array}
           \right]  & \mbox{if } y_i = 0 \\ \\
            \left[
           \begin{array}{ccc}
                0 & 1  & 1 \\
                1 & 0  & 1\\
                1 & 1  & 0
           \end{array}
           \right] & \mbox{if } y_i = 1.
        \end{array}
        \right.
    \end{gather*}
    Let $\M = \A + \B$, with addition over $\F_2$. Then $\M$ is a permutation matrix if and only if $x$ and $y$ are disjoint. Thus, the query complexity is $\Omega(n)$ since $\u^{\T}\A\v$ and $\u^{\T}\B\v$ are both a single bit. 
\end{proof}

\subsection{Doubly Stochastic Matrix}

A non-negative real-valued matrix is {\em doubly stochastic} if all rows and columns sum to one. Similar to permutation matrices, testing if a matrix is doubly stochastic only needs $O(1)$ queries over $\R$.

\begin{theorem} \thmlab{double_stochastic_upper}
    Let $\M \in (\mathbb{R}^+ \cup \{ 0\})^{n\times n}$ be a non-negative real matrix. Then $O(1)$ queries suffices to check whether $\M$ is doubly stochastic with constant probability.
\end{theorem}

\begin{proof}
    The argument is similar to permutation matrix. First, check whether the sum of all entries is $n$ by choosing $\u = \v = (1, 1, \ldots, 1 )^{\T}$ and checking whether $\u^{\T}\M\v = n$. We assume the equality holds. If $\M$ is not doubly stochastic, then some column $c$ (or row $r$) should have sum $>$ 1 and another column $c'$ (or row $r'$ respectively) should have sum $<$ 1. Partition the columns (or rows) into two groups of size $\frac{n}{2}$, then by the same argument as the proof of \thmref{permutation_upper} , the sum of two groups will not be equal with constant probability.
    
\end{proof}

\subsection{Matrix with Negative Entries}

In our previous result for doubly stochastic matrices, we assumed that all the entries are non-negative. This assumption is necessary. If we allow negative entries in a matrix, then even checking whether or not there exists a negative entry requires $\Omega(n^2 / \log n)$ queries.

\begin{theorem} \thmlab{negative}
   Let $\M \in \mathbb{R}^{n\times n}$ be a matrix. Then $\Omega(n^2 / \log n)$ queries are necessary to test if $\M$ contains a negative entry using query vectors with entries in $\{0, \pm 1, \pm 2, \ldots, \pm n^c \}$ for a constant $c$.  
\end{theorem}
\begin{proof}
    We reduce this problem to $\textsc{Disjointness}$.
    Alice holds a bit-string $x$ with size $n^2$, and Bob holds a bit-string $y$ with the same size. Alice and Bob construct $n \times n$ matrices $\A$ and $\B$, where
    \begin{gather*}
        A_{ij} = 
        \left\{
        \begin{array}{rl}
            1, & x_{(i-1)n+j} = 0 \\
            0, &  x_{(i-1)n+j} = 1
        \end{array}
        \right.
\qquad \mbox{ and } \qquad
        B_{ij} = 
        \left\{
        \begin{array}{rl}
            0, & y_{(i-1)n+j} = 0 \\
            -1, &  y_{(i-1)n+j} = 1
        \end{array}
        \right.
    \end{gather*}
Let $\M = \A+\B$, and notice that $\M$ contains negative entries if and only if $x$ and $y$ intersect. If the query complexity is $q(n)$, then by  the $\textsc{Disjointness}$ lower bound, $q(n) \log n = \Omega(n^2)$.
\end{proof}

\section{Graph Problems}
\seclab{graphs}

\subsection{Triangle Detection}

Triangle detection task means that a simple graph $G$ is given in the form of adjacency matrix $\M \in \{0, 1\}^{n\times n}$, where $n$ is the number of vertices in $G$, and we want to decide whether there exists a triangle, i.e. there exists $1 \leq i < j < k \leq n$, such that $M_{ij} = M_{jk} = M_{ki} = 1$. The following theorem shows a lower bound on the number of $\uMv$ queries to detect a triangle. 

\begin{theorem} \thmlab{triangle}
    Given a simple graph $G$ consisting of $n$ vertices in the form of its adjacency matrix $\M \in \{0, 1\}^{n\times n}$, then even with a constant probability, $\Omega(n^2/ \log n)$ queries are necessary to determine whether there exists a triangle in $G$.
\end{theorem}

\begin{proof}
    We reduce this problem to a communication complexity problem, that is \cite{bar2002reductions}, given a graph $G$ with $n$ vertices, where Alice holds some edges of $G$, and Bob holds the remaining edges of $G$, then $\Omega(n^2)$ bits of communication is required to determine whether there exists a triangle in $G$, even in the random case with a constant probability.     
    
    Now suppose the graph is $G$, and its adjacency matrix is $\M \in \{0, 1\}^{n \times n}$. Alice holds some edges represented by the matrix $\X$, and Bob holds the remaining edges represented by the matrix $\Y$. Obviously $\M = \X + \Y$. Then Alice and Bob can communicate by sending $\u^{\T}\X\v$ and $\u^{\T}\Y\v$, and $\u^{\T}\M\v$ can be obtained immediately since $\u^{\T}\M\v = \u^{\T}\X\v + \u^{\T}\Y\v$. Assume $q(n)$ queries can determine whether there exists a triangle, then $q(n)\log n = \Omega(n^2)$. Thus $q(n) = \Omega(n^2 / \log n)$.
    
\end{proof}

\subsection{Deciding if a Graph is a Star}

A {\em star} is a tree where there exists one vertex adjacent to all the other vertices. Given the adjacency matrix of a graph $G$, how many queries do we need to decide whether $G$ is a star?

\begin{theorem} \thmlab{star}
    $\M \in \{0, 1\}^{n\times n}$ is the adjacency matrix of a simple graph $G$. Then $O(1)$ queries suffice to determine whether $G$ is a star with constant probability over $\R$.
\end{theorem}

\begin{proof}
First, check whether $\M$ contains exact $2(n-1)$ ones. If not, $\M$ is obviously not a star. 

Now we assume that $\M$ contains exact $2(n-1)$ ones, which means $G$ contains $(n-1)$ edges. Equally divide the vertices into 2 groups of size $\frac{n}{2}$ randomly and uniformly. We only need to check whether the sum of degrees in one group is $\frac{n}{2}$, and another $\frac{3n}{2} - 2$. If this is true, the algorithm should report that $\G$ is a star. Otherwise, the algorithm reports that $\G$ is not a star.  We prove this algorithm has a constant error probability.

If $G$ is a star, then the sum of degrees in one group is $\frac{n}{2}$, and the other is $\frac{3n}{2} - 2$. If $G$ is not a star, there exists two vertices $v_1$ and $v_2$ with different degrees, which satisfy $|\text{deg}(v_1) - \text{deg}(v_2)| < n-2$. Since $G$ is not a star, then the degree of any vertex can be at most $n-2$. 

\begin{itemize}
    \item If there exists a vertex $v_1$ with degree $n-2$, then when $n$ is large (e.g. $n > 10$), there must exists another vertex $v_2$ with degree 1. Therefore, $$|\text{deg}(v_1) - \text{deg}(v_2)| = (n-2) - 1 = n-3 < n-2.$$
    
    \item If there does not exists a vertex with degree $n-2$, then the degree of all vertices are in $\{0, 1, 2, \ldots , n-3\}$. When $n$ is large enough (e.g. $n>10$), there must exist two vertices $v_1$ and $v_2$ with different degrees, and $$|\text{deg}(v_1) - \text{deg}(v_2)| \leq (n-3) - 0 = n-3 < n-2.$$
    
\end{itemize}
 
Now with probability at least $\frac{1}{2}$, $v_1$ and $v_2$ are in different groups. Without loss of generality, assume that $\text{deg}(v_1) > \text{deg}(v_2)$.  Conditioned on this, we can decompose the random partition procedure into 2 steps. First, we randomly and uniformly partition the other $n-2$ vertices (except $v_1$ and $v_2$) into 2 groups with size $\frac{n}{2} - 1$. Assume that in Group 1 the sum of degrees of these vertices is $s_1$, and in Group 2 the sum is $s_2$. Without loss of generality, assume that $s_1 \leq s_2$. The second step is to place $v_1$ in Group 1, $v_2$ in Group 2, or $v_2$ in Group 2, $v_1$ in Group 1 both with probability $\frac{1}{2}$.

If we have that the following holds simultaneously,
$$\begin{array}{lll}
    s_1+\text{deg}(v_1) = \frac{3n}{2} - 2\text{, } & \quad & s_2+\text{deg}(v_2) = \frac{n}{2}\text{,} \\
   s_1+\text{deg}(v_2) = \frac{n}{2}\text{, } &\quad & s_2+\text{deg}(v_1) = \frac{3n}{2} - 2,
\end{array}$$
then $\text{deg}(v_1) - \text{deg}(v_2) = n-2$, a contradiction. 
Otherwise, if 
$$\begin{array}{lll}
    s_1+\text{deg}(v_1) = \frac{n}{2}\text{, } & \quad & s_2+\text{deg}(v_2) =  \frac{3n}{2} - 2\text{,} \\
   s_1+\text{deg}(v_2) = \frac{n}{2}\text{, } &\quad & s_2+\text{deg}(v_1) = \frac{3n}{2} - 2\text{, }
\end{array}$$
then $\text{deg}(v_1) = \text{deg}(v_2)$, a contradiction.
Also, since $s_1 \leq s_2, \text{deg}(v_1) > \text{deg}(v_2)$ and $\frac{3n}{2} - 2 \geq \frac{n}{2}$, it is impossible that $s_1+\text{deg}(v_2) = \frac{3n}{2} - 2$ and $s_2+\text{deg}(v_1) = \frac{n}{2}$ both hold.
Overall, with probability at least $\frac{1}{4}$, the sums of the two groups will not be $\frac{3n}{2} - 2$ and $\frac{n}{2}$.

\end{proof}

\subsection{Local Graph Queries and Estimating Subgraph Counts}
\seclab{local-graph}

\begin{lemma} \lemlab{local}
    Given the adjacency matrix $\M \in \{0, 1\}^{n \times n}$ of a simple graph $G = (V, E)$, then the following four queries can be implemented by $O(\log n)$ $\uMv$ queries over $\R$:
    \begin{itemize}
        \item $\textbf{Degree}$ $\textbf{query}$ $i$: the degree of vertex $i$. 
        \item $\textbf{Neighbor}$ $\textbf{query}$  $(i, j)$: the \nth{j} neighbor of vertex $i$.
        \item $\textbf{Edge Existence}$ $\textbf{query}$ $(i, j)$: whether the edge $(i, j)$ exists.
        \item $\textbf{Edge-sample}$ $\textbf{query}$: sample an edge $e$ uniformly at random from the edge set $E$.
    \end{itemize}
\end{lemma}
\begin{proof}
    We consider the four queries one by one.
     \begin{itemize}
        \item $\textbf{Degree}$ $\textbf{query}$ $i$: If $\u = (1, 1, \ldots, 1)^{\T}$ and $\v = \e_i$, then $\u^{\T}\M\v$ directly gives the answer.
        \item $\textbf{Neighbor}$ $\textbf{query}$  $(i, j)$: Define $\v^{(k)}$ where $v_l^{(k)} = 1$ when $l \leq k$, and $v_l^{(k)} = 0$ when $l > k$. Then use binary search to determine the \nth{j} neighbor of vertex $i$. First compute $c = \left(\v^{(n/2)}\right)^{\T}\M\e_i$. If $c < j$, then compute $\left(\v^{(3n/2)}\right)^{\T}\M\e_i$. Otherwise, compute $\left(\v^{(n/2)}\right)^{\T}\M\e_i$. By $O(\log n)$ iterations, we can obtain the \nth{j} neighbor of vertex $i$ exactly.
        
        \item $\textbf{Edge Existence}$ $\textbf{query}$ $(i, j)$: Let $\u = \e_i$, and $\v = \e_j$, then $\u^{\T}\M\v$ is the answer.
        \item $\textbf{Edge-sample}$ $\textbf{query}$: Recall that a single query can compute the number of ones in any submatrix. First, determine the number $m$ of ones in $\M$. Then, split $\M$ column-wise into two submatrices $\M_1$ and $\M_2$ with equal size, and compute the number of ones $m_1$ and $m_2$ contained in them. Next, choose $\M_1$ with probability $\frac{m_1}{m}$ or $\M_2$ with probability $\frac{m_2}{m}$, where $m=m_1+m_2$. Assume the chosen matrix is $\M_i$. Recursively perform the same procedure on $\M_i$. After $O(\log n)$ iterations, we obtain a 1 $\times$ 1 matrix with entry 1, which corresponds to a randomly sampled edge. By construction, each edge is chosen with the same probability.
    \end{itemize}
\end{proof}

As one application, we mention the problem of counting subgraphs. Given the adjacency matrix $\M \in \{0, 1\}^{n \times n}$ of a simple graph $G$, we want to estimate the number of occurrences of $H$ in $G$, where $H$ is a given subgraph (such as a triangle). Assadi et. al.~\cite{assadi2018simple} shows that with $$\wt O\left(\text{min}\left\{ m, \frac{m^{\rho(H)}}{\#H} \right\}\right)$$ of the above four standard graph queries, we can obtain a $(1\pm \varepsilon)$-approximation to the number of occurrences of $H$ in $G$ with high probability. Here, $\#H$ is the number of occurrences of $H$ in $G$, $m$ is the number of edges,  and $\rho(H)$ is the fractional edge-cover of $H$. Also, the $\wt O(\cdot)$ notation ignores $\varepsilon$ and $\log n$ terms, as well as the size of graph $H$. By \lemref{local}, the four standard graph queries can be implemented by $O(\log n)$ $\uMv$ queries. Therefore, we derive the following result.

\begin{proposition}
\proplab{counting-subgraphs}
     Given the adjacency matrix $\M \in \{0, 1\}^{n \times n}$ of a simple graph $G$ and an arbitrary small target graph $H$, $\wt O\left(\text{min}\left\{ m, \frac{m^{\rho(H)}}{\#H} \right\}\right)$ $\uMv$ queries suffice to obtain a $(1\pm \varepsilon)$-approximation to the number of occurrences of $H$ in $G$ with high probability.
\end{proposition}

We briefly compare this to work on independent set queries~\cite{bbgm-teuti-19, bbgm-heups-19, dlm-acssw-20}. \propref{counting-subgraphs} achieves a general result for $\uMv$ queries, whereas estimating triangles or other subgraphs with bipartite independent set queries is an open question. Moreover, estimating larger subgraphs seems to require higher-order queries (e.g., tripartite independent set queries). This suggests that, as expected, $\uMv$ queries may be more powerful for a variety of problems.

\section{Conclusion}
\seclab{conclusion}

In this paper, we undertook an exploratory study of a new query model that considers querying a matrix through vector-matrix-vector queries. We provided new algorithms and lower bounds for problems spanning three domains: linear algebra, statistics, and graphs. For many of our results, we showed nearly matching bounds on the query complexity, sometimes up to logarithmic factors. We also demonstrated that many previously studied queries can be viewed as special cases or variants of the $\uMv$ model, and therefore, $\uMv$ queries provide a unified way to study the query complexity of various graph and matrix problems.

In terms of open questions, an interesting direction would be to identify cases where $\uMv$ queries are much more efficient than previously studied models. Some options include: determining the minimum cut more efficiently than cut queries~\cite{rubinstein2018computing, mukhopadhyay2019weighted} or estimating subgraph counts (e.g., triangles) more efficiently than local graph queries~\cite{assadi2017estimating, elrs-actst-17, s-ssaat-15}. Another direction is to identify more problems in linear algebra or statistics where a sublinear or even constant number of $\uMv$ queries suffice. It could also be interesting to study the generalization of our model to $k$-linear forms (i.e., querying a $k$-tensor by specifying $k$ vectors), comparing against $k$-partite independent set queries for counting $k$-cliques~\cite{bbgm-teuti-19, bbgm-heups-19, dlm-acssw-20}.

\paragraph{Acknowledgements} D. Woodruff would like to thank support in part by the Office of Naval Research (ONR) grant N00014-18-1-2562.

\bibliographystyle{alpha}%
\bibliography{edge_est}

\end{document}